\documentclass[12pt,reqno]{amsart}
\usepackage{amsmath}
\usepackage{latexsym}
\usepackage{amsfonts}
\usepackage{amssymb}
\usepackage{color}
\usepackage{bbm,dsfont}
\usepackage{graphicx}
\usepackage{hyperref}

\newtheorem{proposition}{Proposition}
\newtheorem{proposition?}{Proposition?}
\newtheorem{theorem}{Theorem}
\newtheorem{lemma}{Lemma}
\newtheorem{corollary}{Corollary}

\theoremstyle{definition}

\newtheorem{remark}{Remark}
\newtheorem{example}{Example}

\newcommand{\real}{\mathbb R} 
\newcommand{\complex}{\mathbb C} 
\newcommand{\nat}{\mathbb N} 
\newcommand{\half}{\frac{1}{2}} 
\newcommand{\mo}[1]{\left| #1 \right|} 

\newcommand{\hi}{\mathcal{H}} 
\newcommand{\hik}{\mathcal{K}} 
\newcommand{\hv}{\mathcal{V}} 
\newcommand{\lh}{\mathcal{L(H)}} 
\newcommand{\eh}{\mathcal{E(H)}} 
\newcommand{\kb}[2]{|#1\rangle\langle#2|} 
\newcommand{\no}[1]{\left\|#1\right\|} 
\newcommand{\tr}[1]{\textrm{tr}\left[#1\right]} 
\newcommand{\rank}{\mathrm{rank}\,} 

\newcommand{\id}{\mathbbm{1}} 

\newcommand{\fourier}{\mathcal{F}} 

\newcommand{\salg}{\mathcal{F}} 
\newcommand{\ltwo}[1]{L^2(#1)} 

\newcommand{\coex}{\hbox{\hskip0.85mm$\circ\hskip-1.2mm\circ$\hskip0.85mm}}
\newcommand{\coexf}{\mathfrak{c}}

\newcommand{\va}{\mathbf{a}} 
\newcommand{\vb}{\mathbf{b}} 
\newcommand{\vsigma}{\boldsymbol{\sigma}} 

\newcommand{\G}{\mathsf{G}}

\newcommand{\cal}{\mathcal}

\newcommand{\algebra}{\mathcal{A}(P_1,P_2)} 
\newcommand{\spec}[1]{\sigma(#1)} 

\begin{document}
\title[]{Coexistence of effects from an algebra of two projections}

\author[Heinosaari]{Teiko Heinosaari$^\natural$}
\address{$\natural$ Turku Centre for Quantum Physics, Department of Physics and Astronomy, University of Turku}
\email{teiko.heinosaari@utu.fi}

\author[Kiukas]{Jukka Kiukas$^\flat$}
\address{$\flat$ School of Mathematical Sciences, University of Nottingham, University Park,
Nottingham, NG7 2RD, UK}
\email{jukka.kiukas@nottingham.ac.uk}

\author[Reitzner]{Daniel Reitzner$^{\dagger,\clubsuit}$}
\address{$\dagger$ Department of Mathematics, Technische Universit\"at M\"unchen, 85748 Garching, Germany}
\address{$\clubsuit$ Research Center for Quantum Information, Slovak Academy of Sciences, D\'ubravsk\'a cesta 9, 845 11 Bratislava, Slovakia}
\email{daniel.reitzner@tum.de}

\date{\today}
 
\begin{abstract}
The coexistence relation of quantum effects is a fundamental structure, describing those pairs of experimental events that can be implemented in a single setup.
Only in the simplest case of qubit effects an analytic characterization of coexistent pairs is known.
We generalize the qubit coexistence characterization to all pairs of effects in arbitrary dimension that belong to the von Neumann algebra generated by two projections. 
We demonstrate the presented mathematical machinery by several examples, and show that it covers physically relevant classes of effect pairs.
\end{abstract}

\maketitle

\section{Introduction}\label{sec:intro}

An elementary event, i.e., one that can be evaluated by two options 'yes' or 'no', in any physical experiment is of the form 'The recorded measurement outcome belongs to a set $X$.'
In quantum theory these kind of events are mathematically described as effect operators, i.e.,~positive operators bounded by the identity $\id$.
A set of effects is called coexistent if the members of the set can be measured together in a single experimental setup, meaning that they are in the range of a positive operator valued measure \cite{ABQMI85, SEO83}. 

The obvious mathematical problem related to this concept is to characterize those sets of effects that are coexistent.
The first characterization of coexistence for a special type of qubit effect pairs was given in \cite{Busch86}. 
This result was later generalized to all pairs of qubit effects \cite{StReHe08,YuLiLiOh10,BuSc10}. 
The characterization depends heavily on the geometric structure of the set of qubit effects, which is in many respects very special compared to higher dimensional effect spaces.

As we show in this paper, there is a natural setting for a generalization of the qubit coexistence characterization if one considers the special algebraic structure rather than the geometric structure of the set of qubit effects.
The von Neumann algebra of all bounded operators on $\complex^2$ is a factor of type $I_2$. Hence, the simplest choice for a generalization is a type $I_2$ von Neumann algebra, i.e., one which is a direct integral of two-dimensional matrix algebras. In finite dimensions direct integral is just a direct sum, and two matrices belonging to the same type $I_2$-algebra just means that they can be simultaneously diagonalized up to two-by-two blocks. 
This immediately suggests that the qubit characterization of coexistence generalizes to this case. 
What makes such a generalization interesting is the fact that there is a mathematically convenient and often physically relevant example of a type $I_2$ von Neumann algebra: the noncommutative part of the von Neumann algebra generated by two projections.

In this paper we generalize the qubit coexistence result to all pairs of effects belonging to an arbitrary von Neumann algebra generated by a pair of projections. Technically, the characterization is expressed as the positivity of a certain function defined on the spectrum of a fixed central element of the algebra. We wish to emphasize that the dimension of the Hilbert space is irrelevant (and can also be infinite); the result applies to any pair of effects, both belonging to the \emph{same} von Neumann algebra generated by \emph{some} fixed pair of projections.

\emph{Outline.} 
In Sect. \ref{sec:review} we recall the concept of coexistence together with some basic results and, in particular, the characterization of coexistent pairs of qubit effects.
In Sect. \ref{sec:algebra} we review the structure of the algebra of two projections and provide tools that are used in later sections.
In Sect. \ref{sec:main} we present the main result and sketch its proof.
The details of the proof are then given in Sect. \ref{sec:details}.
The use of the characterization in practice is demonstrated in Sect. \ref{sec:scaled}, and the paper concludes with a discussion on the applicability of the characterization.

\emph{Notation.}
In this paper $\hi$ is a fixed complex separable Hilbert space, either finite or infinite dimensional. 
The set of all bounded operators on $\hi$ is denoted by $\lh$. An operator $A\in\lh$ satisfying $0\leq A \leq \id$ is called an \emph{effect}, and we denote by $\eh$ the set of all effects. For each $A\in\eh$, we denote $A^\perp=\id-A$.

\section{Known conditions for coexistence}
\label{sec:review}

Quantum observables are generally described by positive operator valued measures (POVMs) \cite{PSAQT82,OQP97,MLQT12}.
Hence, from the mathematical point of view a quantum observable is a weak-* $\sigma$-additive mapping $\G: \salg\to\eh$ with $\G(\Omega)=\id$, where $\salg$ is a $\sigma$-algebra of subsets of the set $\Omega$ of possible outcomes. In the case of finite outcome set $\Omega$, observable can be identified with a function $x \mapsto \G(x)$ from $\Omega$ to $\eh$ such that $\sum_{x\in \Omega} \G(x) =\id$. 
The probability for the outcome $x\in \Omega$ to occur with system prepared in a state $\varrho$ is then given by $\tr{\varrho\G(x)}$.

\subsection{Basic facts on coexistence}

A set $\mathcal{E}_0\subset\eh$ of effects is \emph{coexistent}, or \emph{consists of coexistent effects}, if there exists an observable $\G: \salg\to\eh$ and events $X_A\in \salg$ for each $A\in\mathcal{E}_0$ such that $A=\G(X_A)$. Note that here the outcome set for $\G$ can be arbitrary.

If two effects $A$ and $B$ are coexistent, we denote $A\coex B$.
It is easy to see that $A \coex B$ if and only if there exist four effects $G_1,G_2,G_3,G_4$ such that
\begin{equation}\label{eq:G}
\begin{gathered}
G_1+G_2=A \, , \quad G_1+G_3=B  \\
G_1+G_2+G_3+G_4=\id \, .
\end{gathered}
\end{equation}
We say that the four-outcome observable consisting of effects satisfying \eqref{eq:G} is a \emph{joint observable} of $A$ and $B$.

A useful equivalent condition for the coexistence of $A$ and $B$ is that the convex set
\begin{equation}\label{coexistenteffects}
\mathcal{C}(A,B):=\{ G \in \lh \mid 0\leq G \leq A, \, A+B -\id \leq G  \leq B\}
\end{equation}
is non-empty. Each $G$ from this set can be identified with $G_1$ from (\ref{eq:G}); the set defining conditions then ensure the existence of $G_2$, $G_3$ and $G_4$. Let us note that the set $\mathcal{C}(A,B)$ is weak-* closed as it is defined by operator inequalities and the cone of positive operators is weak-* closed.
Furthermore, it is subset of $\cal B$, the closed unit ball of $\lh$.
In conclusion, $\mathcal{C}(A,B)$ is a compact convex subset of the compact Polish space $\cal B$.\footnote{The closed unit ball $\cal B$ of $\lh$ is a compact Polish space with respect to the weak-* topology, i.e., it is completely metrizable and separable. 
Indeed, the unit ball of the dual of a separable Banach space is always weak-* compact and metrizable. By compactness, the metric is then complete. Weak-* separability follows easily from the separability of the Hilbert space and the fact that weak-* topology coincides with the weak operator topology on the unit ball.}
This topological aspect leads to a useful observation; all weak-* limits of coexistent effect pairs are coexistent.

\begin{proposition}\label{limitcoexlemma} 
Let $(A_n)$ and $(B_n)$ be two sequences of effects, converging in the weak-* topology to effects $A$ and $B$, respectively. If $A_n\coex B_n$ for each $n\in \nat$, then $A\coex B$.
\end{proposition}

\begin{proof} For each $n$, there exists $G_n\in \mathcal{C}(A_n,B_n)$. 
Since the unit ball $\cal B$ is weak-* compact and metrizable, there exists a subsequence $(G_{n_k})$ converging to some $G\in \cal B$ in the weak-* topology. Since $0\leq G_{n_k}\leq A_{n_k}$, $0\leq G_{n_k}\leq B_{n_k}$, and $A_{n_k}+B_{n_k} -\id \leq G_{n_k}\leq B_{n_k}$ for each $k$, and positivity is preserved in weak-* convergence, it follows that $G\in \mathcal{C}(A,B)$.
\end{proof}

We recall that the \emph{generalized infimum} of two effects $A$ and $B$ is defined as
\begin{align}
A \sqcap B := \half (A + B - \mo{A-B} ) \, .
\end{align}
A recently found \cite{Heinosaari13} simple sufficient condition for the coexistence of $A$ and $B$ is the following:
\begin{equation}
\label{eq:ginf}
\tag{GINF}
(A \sqcap B \geq 0\text{ and }A^\perp \sqcap B^\perp \geq 0)\text{ or }(A \sqcap B^\perp \geq 0\text{ and }A^\perp \sqcap B \geq 0) \, .
\end{equation}
It is easy to see that the following two well-known sufficient conditions for coexistence are special cases of \eqref{eq:ginf}:
\begin{itemize}
\item $AB=BA$ \quad (\emph{commutativity})
\item $A \leq B \quad \textrm{or} \quad B\leq A \quad \textrm{or} \quad A \leq B^\bot \quad \textrm{or} \quad B^\bot \leq A$ \quad  (\emph{comparability})
\end{itemize}
As shown in \cite{Heinosaari13} and further demonstrated later in this paper, \eqref{eq:ginf} covers much wider class of coexistent effects than commutativity and comparability together.

\subsection{Coexistence of two qubit effects}

As explained in the introduction, our characterization of coexistence of two effects belonging to the algebra of two projections is done by reducing the problem to two-dimensional cases.
For this reason we review the complete characterization of coexistence of qubit effects \cite{StReHe08,YuLiLiOh10,BuSc10}. 
In this subsection $\hi=\complex^2$.

Qubit effects can be parametrized by vectors $(\alpha,\va)\in\real\times\real^3$ in the following way:
\begin{align}\label{definition of A}
A =\frac{1}{2}(\alpha I+\va\cdot\vsigma),\qquad \no{\va} \leq \alpha \leq 2-\no{\va} \, .
\end{align}
Here $\boldsymbol{\sigma}\equiv(\sigma_x,\sigma_y,\sigma_z)$ is the vector of the Pauli matrices.
Note that from $\no{\va} \leq \alpha \leq 2-\no{\va}$ follows that $\no{\va}\leq 1$. 
We say that $A$ is \emph{unbiased} if $\alpha=1$; otherwise we call $A$ \emph{biased}.

We are considering the conditions for the coexistence of two qubit effects $A =\frac{1}{2}(\alpha \id +\va\cdot\vsigma)$ and $B =\frac{1}{2}(\beta \id+\vb\cdot\vsigma)$.
To write down the general coexistence condition, we denote 
\[
\langle A|B\rangle := \alpha\beta-\va\cdot\vb,
\]
and define a function $\coexf:\eh\times\eh\to\real$ by
\begin{align*}\label{coexfunc}
\coexf(A,B)&:=\left[ \langle A|A\rangle\langle A^\perp|A^\perp\rangle\langle B|B\rangle\langle B^\perp|B^\perp\rangle \right]^{1/2}\\
&-\langle A|A^\perp\rangle\langle B|B^\perp\rangle+\langle A|B^\perp\rangle\langle A^\perp|B\rangle+\langle A|B\rangle\langle A^\perp|B^\perp\rangle.
\end{align*}
Some immediate properties of the function $\coexf$ are:
\begin{itemize}
\item $\coexf(B,A)=\coexf(A,B)$ (\emph{symmetry})
\item $\coexf(A^\perp,B)=\coexf(A^\perp,B^\perp)=\coexf(A,B)$ (\emph{invariance under complements})
\end{itemize}
However, the function $\coexf$ does not seem to have any direct physical interpretation.

The coexistence relation in $\mathcal{E}(\complex^2)$ can be characterized by a single inequality:

\begin{theorem}[\cite{BuSc10}]\label{qubitcoexthm} 
$A\coex B$ if and only if $\coexf(A,B)\geq 0$.
\end{theorem}

The form of the function $\coexf$ simplifies for some special classes of effects, hence leading to an easier coexistence criterion.
For instance, 
\begin{itemize}
\item if $A,A^\perp,B$ or $B^\perp$ is rank-1 then 
\begin{equation*}
\coexf(A,B)=-\langle A|A^\perp\rangle\langle B|B^\perp\rangle+\langle A|B^\perp\rangle\langle A^\perp|B\rangle+\langle A|B\rangle\langle A^\perp|B^\perp\rangle \, .
\end{equation*}
\item if $A$ and $B$ are unbiased then
\begin{equation*}
\coexf(A,B)=\langle A|A\rangle\langle B|B\rangle-\langle A|A^\perp\rangle\langle B|B^\perp\rangle+\langle A|B^\perp\rangle^2+\langle A|B\rangle^2 \, .
\end{equation*}
\end{itemize}
In the latter case the coexistence criterion was first derived in \cite{Busch86}. 
We also recall from \cite{Heinosaari13} that for two unbiased effects their coexistence is equivalent to \eqref{eq:ginf}. 
However, the coexistence of two arbitrary qubit effects is not equivalent to \eqref{eq:ginf}.

\subsection{From smaller to larger}
\label{sec:smaller2larger}

As explained in the introduction, the idea of the present paper is to extend the qubit coexistence result to higher dimensional Hilbert spaces. 
This in mind, it is appropriate to review some quite direct ways how coexistence results for effects in a smaller Hilbert space can be used to get information on the coexistence of some related effects in a larger Hilbert space, and vice versa.

Let us first consider the passage from finite-dimensional Hilbert spaces to infinite-dimensional.

\begin{proposition}\label{prop:sequence}
Let $\hi$ be an infinite dimensional Hilbert space and $A,B\in\eh$.
Let $(\hv_n)$ be an increasing sequence of closed subspaces of $\hi$ such that $\cup_n \hv_n = \hi$, and for each $n\in\nat$ we denote by $V_n:\hv_n\to \hi$ the canonical isometry. 
Then $A\coex B$ if and only if $V_{n}^*A V_{n}\coex V_{n}^*BV_{n}$ as elements of $\mathcal{E}(\hv_n)$, for each $n$.
\end{proposition}

\begin{proof} If $A\coex B$, when there exists a $G\in \mathcal{C}(A,B)$, then by positivity and linearity, $V_n^*GV_n\in \mathcal{C}(V_{n}^*A V_{n},V_{n}^*BV_{n})$, so that $V_{n}^*A V_{n}\coex V_{n}^*BV_{n}$ (as elements of $\mathcal{E}(\hv_n)$). Concerning the other direction, we first note that the projections $P_n:=V_{n}V_{n}^*$ converge strongly to the identity, so the effects $A_n:=P_nAP_n$ and $B_n:=P_nBP_n$ converge in the weak-* topology to $A$ and $B$, respectively. Assuming that $V_{n}^*A_nV_{n}\coex V_{n}^*B_{n}V_{n}$ (as elements of $\mathcal{E}(\hv_n)$), it follows that $A_n\coex B_n$ (as elements of $\mathcal{E}(\hi)$), because they act trivially outside the subspace. Hence, it follows from Prop. \ref{limitcoexlemma} that $A\coex B$.
\end{proof}

We conclude from Prop. \ref{prop:sequence} that  coexistence in any separable Hilbert space can in principle be reduced to coexistence in its finite-dimensional subspaces. Next, we look at a basic way of obtaining larger finite-dimensional Hilbert spaces from small ones, namely via the tensor product.

Let us consider a Hilbert space $\hi_1 \otimes \hi_2$ and two pairs of effects: $A_1,B_1\in\mathcal{E}(\hi_1)$ and  $A_2,B_2\in\mathcal{E}(\hi_2)$.
Then $A_1\coex B_1$ and $A_2 \coex B_2$ imply that $A_1\otimes A_2 \coex B_1 \otimes B_2$. Indeed, suppose that $i\mapsto G(i)\in\mathcal{E}(\hi_1)$ is a joint observable for $A_1$ and $B_1$ and  $j\mapsto H(j)\in\mathcal{E}(\hi_2)$ for $A_2$ and $B_2$.
Then clearly the observable $(i,j)\mapsto G(i) \otimes H(j)$ contains the effects $A_1\otimes A_2$ and $B_1 \otimes B_2$ in its range.

Perhaps surprisingly, the converse is not true: we may have $A_1\otimes A_2 \coex B_1 \otimes B_2$ even if $A_1\coex B_1$ and $A_2 \coex B_2$ do not hold.
This is demonstrated in the following example.

\begin{example}\label{ex:tensor}
Define four qubit effects: $A_1=A_2=\frac{1}{2\sqrt{2}} (\id + \sigma_z)$ and $B_1=B_2=\frac{1}{2\sqrt{2}} (\id + \sigma_x)$.
From Theorem \ref{qubitcoexthm} follows that $A_1$ and $B_1$ are not coexistent.
However, $A_1\otimes A_2$ and $B_1 \otimes B_2$ are coexistent since $A_1\otimes A_2 + B_1 \otimes B_2 \leq \id$.
\end{example}

Hence, coexistence in a tensor product Hilbert space cannot be characterised in terms of coexistence in the tensor factors, even when the effects are in the product form. Accordingly, the qubit characterisation only provides a sufficient condition for coexistence of two product form effects in a Hilbert space of the form $\hi = \otimes_{i=1}^n \complex^2$.
In contrast to this fact, it is easy to see that coexistence is respected by \emph{direct sum} decompositions.

\begin{proposition}\label{discretedirectsum} Let
$$
\hi = \bigoplus_{i=1}^n \hi_i
$$
denote the direct sum of Hilbert spaces $\hi_i$, and let $A_i,B_i\in \mathcal{E}(\hi_i)$ for each $i=1,\ldots,n$. 
Then the effects
\begin{align*}
A &:=\bigoplus_{i=1}^n A_i, & B &:=\bigoplus_{i=1}^n B_i
\end{align*}
are coexistent if and only if $A_i\coex B_i$ for each $i=1,\ldots,n$. In that case, $A$ and $B$ have a joint observable whose four effects are decomposable in this direct sum.
\end{proposition}

\begin{proof} If $G_i\in\mathcal C(A_i, B_i)$ for each $i$, it is clear that $G:=\oplus_{i=1}^n G_i \in \mathcal C(A, B)$, because positivity is preserved in a direct sum. The effect $G$ is decomposable, so all the effects of the corresponding joint observable given by \eqref{eq:G} are decomposable as well. For the converse part, let $V_i:\hi_i\to \hi$ be the canonical isometry, so that e.g. $V_i^*AV_i =A_i$ for each $i$. Assuming the existence of $G\in\mathcal{C}(A,B)$, it then follows by positivity and linearity that $V_i^*GV_i\in \mathcal{C}(A_i,B_i)$, so that $A_i\coex B_i$. 
\end{proof}

\begin{remark}
 We generalise Prop. \ref{discretedirectsum} to direct integrals in Sect. \ref{sec:details}, Prop. \ref{dirintcoex}. The generalisation is needed for the infinite-dimensional case, which is relevant because of applications to continuous variable systems, in particular, binarizations of the canonical variables. The proof requires some measure theoretic technicalities, which we do not consider essential at this point.
\end{remark}

We conclude from this discussion that direct sums, unlike tensor products, behave neatly under the coexistence relation.
Therefore, direct sums seem more appropriate way to generalize the qubit coexistence characterization to higher dimensions than tensor products. Explicitly, the above proposition has the following consequence: in a Hilbert space of even dimension, the coexistence of any two effects which can be simultaneously diagonalised up to two by two blocks, can be checked using the qubit characterisation. As mentioned in the introduction, an interesting class of such effects is given by the algebra of two projections. 
This is the topic of the next section.

\section{Algebra of two projections}
\label{sec:algebra}

We will utilize the known structure of the algebra generated by the two projections. This algebra has been studied extensively because of its relative simplicity and wide range of applications \cite{PTLO66,Pedersen68,Halmos69,DaKa70,GiKu71,RaSi89,AvSeSi94,Borac95,Vasilevski98,KiWe10}.

\subsection{Definition and basic properties}

Let $P_1$ and $P_2$ be two projections on the Hilbert space $\hi$. These remain fixed throughout the paper. 
We denote by $\algebra$ the von Neumann algebra generated by $P_1$, $P_2$ and the identity operator $\id$ on $\hi$. Hence, $\algebra$ contains all polynomials of $P_1$ and $P_2$, as well as their weak-* limits. In particular, spectral projections of any selfadjoint polynomial are in $\algebra$.

We let $[P_1,P_2]:=P_1P_2-P_2P_1$ denote the commutator of $P_1$ and $P_2$.
The closed subspace $\ker \bigl( [P_1,P_2] \bigr)$ of $\hi$ is called the \emph{commutation domain} of $P_1$ and $P_2$;
if $\psi\in \ker \bigl( [P_1,P_2] \bigr)$, then $P_1P_2\psi=P_2P_1\psi$.
Even more, the restrictions to the commutation domain of any two elements from $\algebra$ commute. 
If $\ker \bigl( [P_1,P_2] \bigr)=\{0\}$, then we say that $P_1$ and $P_2$ are \emph{totally non-commutative}.

An important subspace for our investigation is the support of the commutator, and we denote
\begin{equation}\label{eq:supp}
\hik := \mathrm{supp} \bigl( [P_1,P_2] \bigr) \equiv \ker \bigl( [P_1,P_2] \bigr)^\perp \, .
\end{equation}
Then $\hik^\perp=\ker \bigl( [P_1,P_2] \bigr)$ and
$$
\hi=\hik\oplus \hik^\perp.
$$
Clearly, $\hik$ and $\hik^\perp$ are invariant subspaces for both $P_1$ and $P_2$. 

For the following developments, it is useful to identify some specific central elements of $\algebra$, i.e., operators that commute with all elements in $\algebra$.
First, we denote 
\begin{align}
\label{eq:centralC}
C  :=\id -(P_1-P_2)^2
\end{align}
and observe two alternative ways to write it:
\begin{align}\label{eq:centralC-2}
C  = P_1P_2P_1 + P_1^\perp P_2^\perp P_1^\perp    = P_2P_1P_2 + P_2^\perp P_1^\perp P_2^\perp  \, .  
\end{align}
It is clear that $C \in \algebra$ and it is also easy to see by a direct computation that $C$ commutes with $P_1$ and $P_2$. 
Since 
\begin{align*}
-\id \leq -P_2\leq P_1-P_2\leq P_1\leq \id \, , 
\end{align*}
it follows that $C$ is an effect.

Two more useful central elements, denoted by $P_{\mathcal{K}}$ and $P_{\mathcal{K}^\perp}$, are the projections onto $\hik$ and $\hik^\perp$, respectively.
Then $P_{\mathcal{K}^\perp}=P_{\mathcal{K}}^\perp$.
Since $\hik$ and $\hik^\perp$ are invariant subspaces for both $P_1$ and $P_2$, it follows that $P_{\mathcal{K}}$ and $P_{\mathcal{K}^\perp}$ commute with $P_1$ and $P_2$.
Moreover, $P_{\mathcal{K}^\perp}$ is a spectral projection of the operator $-[P_1,P_2]^2=CC^\perp$ corresponding to the eigenvalue $0$.
Since $CC^\perp\in\algebra$, we also have $P_{\mathcal{K}^\perp}\in \algebra$.
In conclusion, $P_{\mathcal{K}}$ and $P_{\mathcal{K}^\perp}$ are central projections of $\algebra$.

Since the elements of $\algebra$ act trivially in the commutation domain, the nontrivial structure is contained in the von Neumann algebra $P_{\mathcal K}\algebra$. This is always type $I_2$, and its center is generated by a single element; this element can be chosen to be $P_{\mathcal{K}}C$ or a suitable function thereof (see e.g. \cite{Borac95}). Hence, $P_{\mathcal K}\algebra$ is isomorphic to $L^\infty(I,\mu,\mathcal L(\complex^2))$, where $I$ is some bounded interval of the real line, and $\mu$ is given by the spectral representation of the central element. It is often convenient to obtain a concrete Hilbert space representation for this isomorphism, and we will do this next.

\subsection{Extracting the qubit}

We now review the way of representing the nontrivial part $P_\hik\algebra$ of $\algebra$ as a tensor product of a commutative part, generated by a single central element, and one qubit algebra.

This representation can be explicitly realized on the Hilbert space level; the starting point is to restrict to the subspace $\hik$, defined in \eqref{eq:supp}. 
The decomposition will be done asymmetrically with respect to $P_1$ and $P_2$, but one can naturally change their roles.
In practice it is often more convenient to choose the projection with a smaller rank as the starting point.

We first decompose
\begin{equation*}
\hik=\hik_0\oplus \hik_1,
\end{equation*}
where 
\begin{align*}
\hik_0 &:=\{\varphi \in \hik\mid P_1\varphi=\varphi \}, & \hik_1 &:= \{\varphi \in \hik\mid P_1\varphi=0 \}.
\end{align*}
Next we look at the compression of the central element $C$ relative to this decomposition. Since $P_{\mathcal{K}}$ and $P_1$ commute, the projection onto $\hik_0$ is just $P_{\mathcal{K}}P_1$, and the projection onto $\hik_1$ is $P_{\mathcal{K}}P_1^\perp$. 
The central element $C$ commutes with these two operators, so we can compress it to each invariant subspace (see \eqref{eq:centralC-2}): 
\begin{align*}
C|_{\hik_0} &= P_1P_2P_1|_{\hik_0}\, ,  & C|_{\hik_1} &=P_1^\perp P_2^\perp P_1^\perp |_{\hik_1} \, .
\end{align*}
The core result in the ``theory of two projections'' (see the above references) says that these two selfadjoint operators are unitarily equivalent, with eigenvalues (if any) being in $(0,1)$ and the entire spectrum is some closed subset of $[0,1]$. 
This unitary equivalence implies that we can identify
\begin{equation*}
\hik = \hik_0\oplus \hik_1 \simeq \hik_0\otimes \complex^2 \, , 
\end{equation*}
in such a way that
\begin{equation*}
C|_{\hik} \simeq H\otimes \id_{\complex^2} \, , \quad \text{ with } H:=C|_{\hik_0} \, .
\end{equation*}
In order to make the connection to the two-dimensional case more explicit, we switch to the \emph{angle operator}
$\Theta$, defined as 
\begin{equation*}
\Theta := \arccos(2H-\id_{\hik_0}) \, . 
\end{equation*}
Since $h\mapsto \arccos(2h-1)$ maps the interval $[0,1]$ bijectively onto $[0,\pi]$, the spectrum $\spec{\Theta}$ of $\Theta$ is a closed subset of $[0,\pi]$, and neither of the endpoints are eigenvalues of $\Theta$.

For our purposes, the main point of this representation is that the two projections 
have the following simple forms:
\begin{equation}\label{eq:C2representationMT}
\begin{split}
P_1|_{\hik} &=\id\otimes \frac 12(\id_{\complex^2} +\sigma_z) \, ,  \\  
P_2 |_{\hik} &= \frac 12 \left(\id\otimes \id_{\complex^2}+\sin \Theta\otimes \sigma_x+\cos\Theta\otimes \sigma_z\right) \, .
\end{split}
\end{equation}
The second factor in $P_1|_{\hik}$ is a qubit projection.
A suitable representation of $P_2|_{\hik}$ depends on the spectrum of $\Theta$; this is treated in the following subsection.
 
\subsection{Direct integral decomposition of $P_\hik\algebra$}

In order to use the decomposition from above in explicit computations, it is convenient to use the spectral theorem to diagonalize the selfadjoint operator $\Theta=\arccos((2P_1P_2P_1-\id)|_{\hik_0})$.

In order to make the discussion more accessible, we consider first the simple case where $\hik_0$ is finite dimensional, with dimension $d$.
Then the spectrum $\spec{\Theta}$ is just the collection of (at most $d$) eigenvalues of $\Theta$, and we can decompose $\hik_0$ into a direct sum of the corresponding eigenspaces $\hi_{\theta}$, $\theta\in \spec{\Theta}$. Then $\hik\simeq \hik_0\otimes \complex^2$ decomposes as
$$
\hik\simeq \bigoplus_{\theta\in \spec{\Theta}} \hi_{\theta}\otimes \complex^2,
$$
and $\Theta$ acts as multiplication by $\theta \id_{\hi_\theta}\otimes \id_{\complex^2}$ in each component.

From \eqref{eq:C2representationMT} we immediately see that the projections decompose in the form
$$
P_i|_{\hik}=\bigoplus_{\theta\in \spec{\Theta}} \id_{\hi_\theta}\otimes M_{P_i}(\theta),
$$
where
\begin{equation}
\begin{split}
M_{P_1}(\theta) &=\frac{1}{2}(\id_{\complex^2}+\sigma_z),\\ M_{P_2}(\theta)&=\frac{1}{2}(\id_{\complex^2}+\sin\theta \sigma_x+\cos \theta\sigma_z).
\end{split}
\label{eq:C2representation}
\end{equation}
This just means that $\hik$ has a basis $\{e_i\}_{i=1}^{2d}$, in which both projections are diagonal up to two-by-two blocks, each block occurring the number of times given by the degeneracy of the corresponding eigenvalue $\theta\in \spec{\Theta}$.

Since none of the eigenvalues $\theta$ equals $0$ or $\pi$ (so that $\sin\theta$ is nonzero), it is clear that in each block, an arbitrary two-by-two matrix can be generated as a suitable linear combination of products of $M_{P_1}(\theta)$ and $M_{P_2}(\theta)$ and $\id_{\complex^2}$. Since the eigenprojections of $\Theta$ are in $\algebra$, this just means that each element of the algebra $P_\hik\algebra$ is of the form
\begin{equation}\label{directsumdecomp}
A|_{\hik}=\bigoplus_{\theta\in \spec{\Theta}} \id_{\hi_\theta}\otimes M_{A}(\theta),
\end{equation}
where each $M_A(\theta)$ is an arbitrary two-by-two matrix. In other words, an effect is an element of  $P_\hik\algebra$, if and only if its matrix in the above mentioned basis $\{ e_i\}$ is block diagonal with two-by-two blocks, such that the blocks associated to the same eigenvalue $\theta$ are equal. 

In the infinite-dimensional case, the spectrum of $\Theta$ can have a continuous part and there might not even be any eigenvalues at all. 
Then the appropriate analogue of the direct sum is the direct integral, and we can diagonalize $\Theta$ in the sense of the spectral theorem (see Prop. \ref{spectraltheorem} in Sect. \ref{sec:details}): For each element $\theta\in \spec{\Theta}$ we have a multiplicity space $\hi_\theta$, and there exists a Borel measure $\mu$ on $[0,\pi]$ with support $\spec{\Theta}$, such that $\hik$ is the direct integral space
\begin{equation}\label{dirintHrep}
\hik \simeq \int_{\spec{\Theta}}^{\oplus} \hi_\theta\otimes \complex^2\, d\mu(\theta) \, , 
\end{equation}
where $\Theta$ is the diagonal multiplication operator $\Theta(\theta) = \theta \id_{\hi_\theta}\otimes \id_{\complex^2}$.

In the analogue of the decomposition \eqref{directsumdecomp}, one has to require measurability for the function $\theta\mapsto M_A(\theta)$. We denote by $L^\infty([0,\pi],\mu,\mathcal{L}(\complex^2))$ the set of equivalence classes of $\mu$-essentially bounded $\mu$-measurable $\mathcal{L}(\complex^2)$-valued functions. The following Prop. characterises the algebra $P_\hik \algebra$ (see e.g. \cite{Borac95} and the references therein):

\begin{proposition}\label{structureprop}
There exists a unique von Neumann algebra isomorphism
$$P_\hik\algebra \ni A\mapsto M_A\in L^\infty([0,\pi],\mu,\mathcal{L}(\complex^2)),$$
such that \eqref{eq:C2representation} holds. Moreover, we have
\begin{equation}\label{decomposition}
A = \int_{[0,\pi]} \id_{\hi_\theta} \otimes M_A(\theta)\, d\mu(\theta)
\end{equation}
for all $A\in P_\hik\algebra$.
\end{proposition}

In particular, $P_{\mathcal{K}}\algebra$ is the von Neumann subalgebra of those decomposable operators \footnote{See Sect. \ref{dirintsec} for definition} $A$ for which $A(\theta)$ acts trivially on each multiplicity space $\hi_\theta$.

Note that, as a function, $M_A$ is determined up to $\mu$-null sets only. However, in many instances of the coexistence problem only the $C^*$-algebra generated by the projections is needed. In the finite-dimensional case, this is just the same as the von Neumann algebra, but it is in general smaller. The following result characterizes this algebra (see e.g. \cite{Pedersen68}).

\begin{proposition} $A\in P_{\mathcal{K}}\algebra$ is the compression of an element of the $C^*$-algebra generated by $P_1$, $P_2$, and $\id$, if and only if
\begin{itemize}
\item[(i)] $M_A$ is continuous on $\spec{\Theta}$, and
\item[(ii)] $M_A(\theta)$ is diagonal for $\theta\in \{0,\pi\}\cap\spec{\Theta}$.
\end{itemize}
\end{proposition}

Being a ``compression" here means that $A=P_{\mathcal{K}}A'$ with $A'$ in the $C^*$-algebra in question. (Note that $P_{\mathcal{K}}$ is not, in general, an element of that algebra.) In the full characterization of the $C^*$-algebra, boundary conditions in (ii) must be related to the structure of the commutation domain; for details, see e.g. \cite{Pedersen68}. Note also that condition (ii) is void in the finite-dimensional case, because $0$ and $\pi$ are not eigenvalues of $\Theta$.

In particular, if $A\in P_\hik\algebra$ is selfadjoint, and $f:\real\to \complex$ is bounded and measurable, then $M_A(\theta)$ is selfadjoint for $\mu$-almost all $\theta$, and
\begin{equation}\label{functions}
M_{f(A)}(\theta)=f(M_A(\theta)) \text{ for $\mu$-almost all $\theta\in \spec{\Theta}$.}
\end{equation}
This observation, together with \eqref{eq:C2representation}, allows one to construct the function $M_A$ for any element $A\in \algebra$. If $A$ is a polynomial of $P_1$ and $P_2$, this can be done explicitly by just multiplying matrices and taking linear combinations.

Since we are dealing mostly with positive operators, it is also important to note the positivity condition (see also Sect. \ref{sec:details}): for $A\in P_\hik\algebra$,
\begin{equation}\label{poscond}
A\geq 0 \text{ if and only if } M_{A}(\theta)\geq 0 \text{ for $\mu$-almost all $\theta\in \spec{\Theta}$.}
\end{equation}
In the case where $A$ is in the $C^*$-algebra and $f$ is continuous, we can omit ''$\mu$-almost'' in \eqref{functions} and \eqref{poscond}, because $\spec{\Theta}$ is the support of $\mu$.

\begin{remark}

Let $\hi=\complex^2$.
In the Bloch vector parametrization described in Sect. \ref{sec:review}, two projections take the form
\begin{align*}
P_1&=\frac 12(\id+{\bf n}_1\cdot \sigma), & P_2&=\frac 12(\id+{\bf n}_2\cdot \sigma),
\end{align*}
where ${\bf n}_1$ and ${\bf n}_2$ are unit vectors in $\real^3$.
Leaving out the commutative case, when ${\bf n}_1=\pm {\bf n}_2$ and  $P_1$ and $P_2$ commute, let $\theta\in (0,\pi)$ be the angle between the vectors ${\bf n}_1$ and ${\bf n}_2$. 
Now $\hi=\hik$, and the angle operator $\Theta$ acts on a one-dimensional subspace and is thus determined by its only eigenvalue, which is $\theta$. Thus, changing to the representation \eqref{dirintHrep} is simply done via a suitable coordinate transformation (i.e. an $SU(2)$ unitary transformation in $\complex^2$), which makes the projections
\begin{equation}
\label{eq:scaledMQ}
P_1=\frac{1}{2}(\id+\sigma_z),\qquad P_2=\frac{1}{2}(\id+\sin\theta \sigma_x+\cos \theta\sigma_z),
\end{equation}
as in \eqref{eq:C2representation}.
Summarizing this case, we have $\spec{\Theta} = \{\theta\}$, $\hi_{\theta}= \complex$, and $\mu=\delta_{\theta}$ (the point measure at $\theta$).
\end{remark}

\section{Main result}\label{sec:main}

We now proceed to the main result of the paper, namely the characterization of coexistence of a pair of effects from $\algebra$.

Clearly, if an element $A\in \algebra$ is an effect, we have $0\leq M_A(\theta)\leq \id_{\complex^2}$ for $\mu$-almost all $\theta\in [0,\pi]$ in the decomposition \eqref{decomposition}. Conversely, a function $M\in L^\infty([0,\pi],\mu, \complex^2)$ with $0\leq M(\theta)\leq \id_{\complex^2}$ for $\mu$-almost all $\theta\in [0,\pi]$ defines a unique effect in $P_{\mathcal{K}}\algebra$.

\subsection{Finite-dimensional case}
We begin with the finite-dimensional case, which is a consequence of Prop. \ref{discretedirectsum} and the qubit characterisation Theorem \ref{qubitcoexthm}.

\begin{proposition}\label{finitedimcase}
Suppose $\hi$ is finite-dimensional, and let $A,B\in \algebra$ be two effects. Then $A\coex B$ if and only if 
\begin{equation}
\coexf(M_{A}(\theta),M_{B}(\theta))\geq 0\text{ for all } \theta\in \spec{\Theta} \, .
\end{equation}
In that case, $A$ and $B$ have a joint observable whose effects are contained in $\algebra$.
\end{proposition}
\begin{proof}
Since $P_{\mathcal{K}^\perp}A$ and $P_{\mathcal{K}^\perp}B$ commute, they are always coexistent in the subspace $\hik^\perp$. Hence, according to Prop. \ref{discretedirectsum} and the decomposition \ref{directsumdecomp}, we have $A\coex B$ if and only if $\id_{\hi_\theta}\otimes M_A(\theta)\coex \id_{\hi_\theta}\otimes M_B(\theta)$ for each $\theta$ in the finite set $\spec{\Theta}$. But this is equivalent to the coexistence of the qubit effects $M_A(\theta)$ and $M_B(\theta)$ for each $\theta$. To see this, fix a unit vector $\eta_\theta\in \hi_\theta$ for each $\theta$, and define $V_\theta:\complex^2\to \hi_\theta\otimes \complex^2$ via $V_\theta\varphi=\eta_\theta\otimes \varphi$. Then e.g. $M_A(\theta)=V_\theta^*(\id_{\hi_\theta}\otimes M_A(\theta))V_\theta$, from which the claim clearly follows.

In the case where $A\coex B$, we can pick $G_\theta\in \mathcal C(M_A(\theta), M_B(\theta))$ for each $\theta$; then $G:=\id_{\hi_\theta}\otimes G_\theta\in \mathcal C(A,B)$, and therefore defines a joint observable for $A$ and $B$ with effects in $\algebra$. This completes the proof. 
\end{proof}

We demonstrate the use of this result in the following example.

\begin{example}
Let $\{\psi_1,\psi_2,\psi_3\}$ be an orthonormal basis for $\hi=\complex^3$ and $\varphi=\frac{1}{\sqrt{3}}(\psi_1+\psi_2+\psi_3)$.
We denote $P_1=\kb{\psi_1}{\psi_1}$ and $P_2=\kb{\varphi}{\varphi}$.
We then have $\ker \bigl( [P_1,P_2] \bigr)=\complex (\psi_2-\psi_3)$ and $\hik$ is the two dimensional subspace spanned by the vectors $\psi_1$ and $\varphi$.
Now $\hik_0=\complex \psi_1$ is one-dimensional, and $$H=P_1P_2P_1|_{\hik_0}=|\langle \psi_1|\varphi\rangle|^2\id_{\mathbb{C}} = \frac{1}{3}\id_\mathbb{C},$$ so that the spectrum of operator $\Theta$ is $\spec{\Theta}=\{\theta\}$ with $\theta=\arccos(-1/3)$.

Let us consider two 1-parameter families of effects,
\begin{align*}
A_s  = s (P_1+P_2) \, , \qquad B_t = t (P_1+P_2^\perp) \, , 
\end{align*}
where the requirements $0\leq A_s \leq \id$ and $0\leq B_t \leq \id$ mean that 
\begin{align*}
0\leq s \leq s_{max} \equiv (3-\sqrt{3})/2 \approx 0.63  \, , \qquad 0\leq t \leq t_{max} \equiv 3-\sqrt{6} \approx 0.55 \, .
\end{align*}
Since $[A_s,B_t]=-2st[P_1,P_2]$, two effects $A_s$ and $B_t$ do not commute unless $s=0$ or $t=0$.
But by using the sufficient condition \eqref{eq:ginf} we find a large region of coexistence drawn in Fig.~\ref{fig:dim3}.

\begin{figure}
\begin{center}
\includegraphics{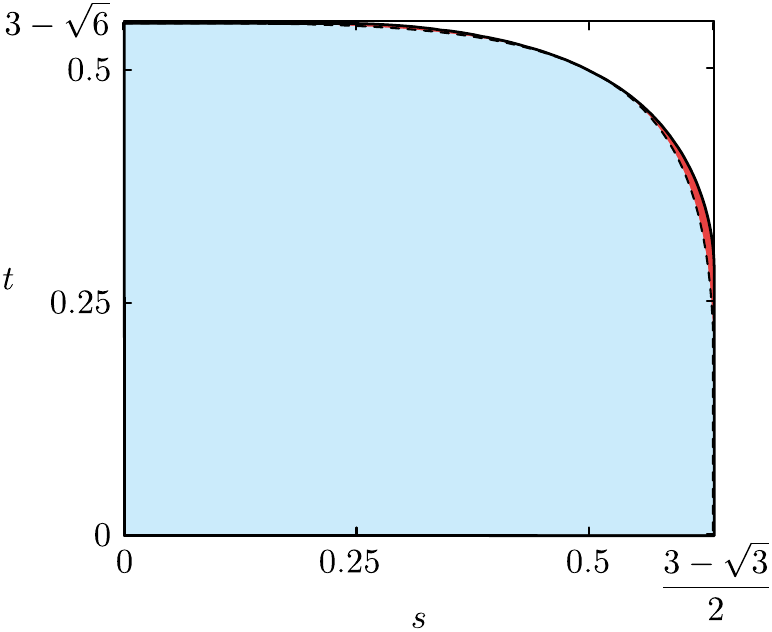}
\end{center}
\caption{Effects $A_s$ and $B_t$ are coexistence when $(s,t)$ lies in the colored region. 
Considering the sufficient condition \eqref{eq:ginf} one gets the region (light blue) of $s$ and $t$ almost covering the whole coexistence region bounded by a fourth order curve (dark red).}
\label{fig:dim3}
\end{figure}
Since \eqref{eq:ginf} provides only a sufficient condition for coexistence, the region of coexistent pairs $A_s$ and $B_t$ can be larger.
The results of this paper provide a complete characterization of coexistence for $A_s$ and $B_t$.  
Using \eqref{eq:C2representation} we get
\begin{align}
M_{A_s}(\theta) &= \half\left[2s\id+s\sin\theta\sigma_x+s(1+\cos\theta)\sigma_z\right],\\
M_{B_t}(\theta) &= \half\left[2t\id-t\sin\theta\sigma_x+t(1-\cos\theta)\sigma_z\right].
\end{align}
Inserting into $\coexf(M_{A_s}(\theta),M_{B_t}(\theta))\geq 0$ we find that the necessary and sufficient condition for the coexistence of $A_s$ and $B_t$ is
\begin{equation}
\label{eq:dim3whole}
\sqrt{2(2s^2-6s+3)(t^2-6t+3)}+16st-12s-15t+9 \geq0 \, .
\end{equation}
The bracketed terms under the square root are positive exactly for $0\leq s\leq s_{max}$ and $0\leq t\leq t_{max}$. 
Considering whole \eqref{eq:dim3whole} one gets a fourth-order curve that is restrictive also only within this region; see Fig.~\ref{fig:dim3}.
We see that \eqref{eq:ginf} covers almost all coexistent pairs, but there is a small region of $(s,t)$ such that the effects $A_s$ and $B_t$ are coexistent even if they do not fulfill \eqref{eq:ginf}.
\end{example}

\subsection{General case} In the infinite-dimensional case, where $\spec{\Theta}$ can have continuous part, the situation is complicated by various measure-theoretic subtleties. Most importantly, the selection of the effects $G_\theta$ in the above proof must be done in a measurable way, which requires nontrivial use of the topological properties of the unit ball of the set of bounded operators. So as not to confuse a reader unfamiliar with infinite-dimensional spaces, we postpone the general proof to Sect. \ref{sec:details}, and present here only the result:

\begin{proposition}\label{prop:compatibilityvNm}
Let $A,B\in \algebra$ be two effects. Then $A\coex B$ if and only if 
\begin{equation}
\coexf(M_{A}(\theta),M_{B}(\theta))\geq 0\text{ for $\mu$-almost all } \theta\in [0,\pi] \, .
\end{equation}
In that case, $A$ and $B$ have a joint observable whose effects are contained in $\algebra$.
\end{proposition}
\begin{proof}
As in the finite-dimensional case, the commutation domain does not play a role. The claim therefore follows from \eqref{decomposition}, Theorem \ref{qubitcoexthm} and Prop. \ref{dirintcoextensor} presented in Sect. \ref{sec:details}.
\end{proof}

\begin{corollary}\label{prop:compatibility}
If $A,B\in \algebra$ belong to the $C^*$-algebra generated by the two projections $P_1$ and $P_2$, then $A\coex B$ if and only if
\begin{equation}
\coexf(M_{A}(\theta),M_{B}(\theta))\geq 0\text{ for all } \theta\in \spec{\Theta} \, .
\end{equation}
\end{corollary}

\begin{proof}
Now the matrix-valued functions $M_A$ and $M_B$ are continuous on $\spec{\Theta}$.
Since $\coexf$ is a continuous function of a pair of qubit effects, the claim follows because $\spec{\Theta}$ is the support of the measure $\mu$.
\end{proof}

\section{Coexistence of projections multiplied by central element}\label{sec:scaled}

Given two noncommutative projections $P_1$ and $P_2$, one can ask how to "smooth" them into coexistent effects without destroying their properties too much. 
Alternatively, we may want to know how much noise $P_1$ and $P_2$ can tolerate and still be non-coexistent.
 An appealing way to study this question is suggested by the structure of the algebra $\algebra$: multiply the projections by suitable central elements. This idea covers a number of interesting special cases.

We now consider effects $A=f_1(C)P_1$, $B=f_2(C)P_2$, where $f_i:[0,1]\to [0,1]$ are measurable functions and $C$ is the previously defined central element $C=\id -(P_1-P_2)^2$. Before characterising the coexistence, we show that the sufficient condition (GINF) is, interestingly, equivalent to coexistence in this setting. 
We start with a simple result.

\begin{proposition}\label{prop:rank1} 
Let $P_1,P_2$ be projections and assume that $\rank P_1=1$.
The effects $sP_1$ and $tP_2$ are coexistent if and only if they commute or $sP_1+tP_2\leq \id$.
\end{proposition}

\begin{proof}
Since commutativity and comparability are always sufficient conditions for coexistence, we only need to show that whenever $sP_1$ and $tP_2$ are coexistent, then they commute or $sP_1+tP_2\leq \id$.
Hence, suppose that $sP_1$ and $tP_2$ do not commute (in particular $s\neq 0 \neq t$) but are coexistent.
By (\ref{coexistenteffects}) there exists $G\in\eh$ such that $G\leq sP_1$, $G\leq tP_2$ and $\id+G \geq sP_1+tP_2$. 
Since $\rank sP_1=1$, we have $G=g_1 sP_1$ for some real number $g_1\geq 0$.
But then $\frac{g_1 s}{t} P_1 \leq P_2$, implying that $[\frac{g_1 s}{t} P_1,P_2]=0$.
Since $sP_1$ and $tP_2$ do not commute, it follows that $g_1=0$ and therefore $G=0$.
Hence, $sP_1+tP_2\leq \id$.
\end{proof}

Prop. \ref{prop:rank1} is not valid if the assumption $\rank P_1=1$ is removed.
For instance, let $\hi=\complex^3$ and choose 
\begin{equation*}
P_1=\id_\complex \oplus P_x \, , \quad  P_2=\id_\complex \oplus P_y \, , 
\end{equation*}
where $P_j = \half (\id + \sigma_j)$ for $j=x,y$.
Using either Theorem \ref{qubitcoexthm} or Prop.~\ref{prop:rank1}, we see that $s P_x$ and $s P_y$ are coexistent if and only if $0 \leq s \leq 2-\sqrt{2}$.
It follows that also $sP_1$ and $sP_2$ are coexistent if and only if $0 \leq s \leq 2-\sqrt{2}$ (see the note on direct sums in Subsec. \ref{sec:smaller2larger}).
But the effects $(2-\sqrt{2}) P_1$ and $(2-\sqrt{2}) P_2$ do neither commute nor satisfy $sP_1+sP_2\leq \id$ ($s$ would have to be smaller than $1/2$).
They do, however, satisfy \eqref{eq:ginf}.
The valid generalization of Prop. \eqref{prop:rank1} is the following.

\begin{proposition}\label{prop:rank1gen} The effects $f_1(C)P_1$ and $f_2(C)P_2$ are coexistent if and only if they satisfy \eqref{eq:ginf}.
\end{proposition}

\begin{proof} 
We first make a general observation: \eqref{eq:ginf} holds for two effects $A,B\in \algebra$ if and only if \eqref{eq:ginf} holds for the qubit effects $M_{A}(\theta)$ and $M_{B}(\theta)$ for $\mu$-almost all $\theta\in \spec{\Theta}$. 
Namely, since \eqref{eq:ginf} consists of positivity conditions, we can check them separately in the $\hik$ and $\hik^\perp$. Using \eqref{functions} and \eqref{poscond} takes care of the former, and for the latter we can use the fact (see e.g. \cite{Halmos69}) that the commutation domain splits into four mutually orthogonal subspaces, and the corresponding restrictions of any element $A\in \algebra$ is a multiple of the identity. Since all the positivity conditions in \eqref{eq:ginf} hold trivially for numbers between zero and one, the coexistence in $\hik^\perp$ is trivially valid.

Let us then consider the effects $A=f_1(C)P_1$ and $B=f_2(C)P_2$.
For each $\theta$, the qubit effects $M_A(\theta)$ and $M_B(\theta)$ are multiples of the rank-1 qubit projections given in \eqref{eq:C2representation}, because $M_{f_i(C)}(\theta)=f_i(\cos^2(\theta/2))\id_{\complex^2}$.
By Prop.~\ref{prop:rank1} they are coexistent if and only if they commute or $M_A(\theta)+M_B(\theta)\leq \id$; both are special cases of \eqref{eq:ginf}. 
(Actually, the commutative case only occurs for $\theta=0$ or $\pi$, but this set has measure zero.) Therefore, by Corollary \ref{prop:compatibility} $A$ and $B$ are coexistent if and only if \eqref{eq:ginf} holds for all $M_{A}(\theta)$ and $M_{B}(\theta)$.
By our earlier observation, this means that $A$ and $B$ satisfy \eqref{eq:ginf}.
\end{proof}

Although checking the positivity conditions in \eqref{eq:ginf} completely characterizes coexistence in the present case, it may be easier to use Corollary~\ref{prop:compatibility} to decide coexistence of a pair of effects. 
This leads to the following result.

\begin{proposition}\label{prop:scaling}
Let $P_1$ and $P_2$ be two non-commuting projections and $f_1,f_2: [0,1]\to[0,1]$ measurable functions.
The effects $f_1(C)P_1$ and $f_2(C)P_2$ are coexistent if and only if
\begin{equation}
\label{eq:funcP}
\inf_{h\in \spec{H}}\left(\frac{1-f_1(h)}{f_1(h)} \frac{1-f_2(h)}{f_2(h)}-h\right) \geq 0.
\end{equation}
In particular, for $s,t\in (0,1]$ the effects $sP_1$ and $tP_2$ are coexistent if and only if
\begin{equation}
\label{eq:scaledP}
\|H\| \leq \frac{1-s}{s} \frac{1-t}{t} \, , 
\end{equation}
and the effects $sP_1P_2P_1$ and $tP_2P_1P_2$ are coexistent if and only if
\begin{equation}
\label{eq:scaledPPP}
\inf_{h\in \spec{H}}\left(\frac{1-sh}{sh} \frac{1-th}{th}-h\right) \geq 0.
\end{equation}
\end{proposition}

\begin{proof} 
Fix $\theta\in\sigma(\Theta)$, and denote $\alpha_i(\theta)= f_i(\cos^2(\theta/2))$, $i=1,2$. Now 
\begin{equation}
\label{eq:scaledM}
\begin{split}
M_{f_1(C)P_1}(\theta) &= \frac{\alpha_1(\theta)}{2}(\id+\sigma_z),\\
M_{f_2(C)P_2}(\theta) &= \frac{\alpha_2(\theta)}{2}(\id+\sin\theta\sigma_x+\cos\theta\sigma_z).
\end{split}
\end{equation}
A direct computation shows that
\begin{multline*}
\coexf(M_{f_1(C)P_1}(\theta),M_{f_2(C)P_2}(\theta))=
2\alpha_1(\theta)\alpha_2(\theta)(1-\cos\theta)\\
\times\left[\alpha_1(\theta)\alpha_2(\theta)(1-\cos\theta)+2(1-\alpha_1(\theta)-\alpha_2(\theta))\right]
\end{multline*}
By Corollary \ref{prop:compatibility}, we have $f_1(C)P_1\coex f_2(C)P_2$ if and only if this expression is nonnegative for all $\theta\in \spec{\Theta}$. After some algebraic manipulation we find that the inequality is equivalent to
\begin{equation}\label{scaledcondition}
\cos^2\frac{\theta}{2}\leq\frac{1-\alpha_1(\theta)}{\alpha_1(\theta)}\frac{1-\alpha_2(\theta)}{\alpha_2(\theta)},
\end{equation}
from which (\ref{eq:funcP}) follows, because $H=\cos^2\frac{\Theta}{2}$.
\end{proof}

\begin{example}\label{QP}
Here we look at an infinite-dimensional case, the scaled position and momentum projections corresponding to half-lines. Let $\hi=\ltwo{\real}$ and let $Q$ and $P$ be the usual position and momentum operators on the real line.
One of the simplest binarizations is the dilation invariant one: $P_1=\chi_{[0,\infty)}(Q)$, and $P_2=\chi_{[0,\infty)}(P)$. (One can also divide at some point other than $0$; the resulting pair of projections is unitarily equivalent to this one via a suitable Weyl translation.)

The two projections $P_1$ and $P_2$ are totally non-commutative \cite{BuLa86}, and the dilation invariance of $P_1$ and $P_2$ can be used to show that they are unitarily equivalent to the decomposable operators
\begin{equation}\label{eq:scaledMQP}
\begin{split}
P_1(v) &=\frac{1}{2}(\id+\sigma_z) \, ,\\
P_2(v) &=\frac{1}{2}[\id+{\rm sech}(\pi v/2)\sigma_x+\tanh(\pi v/2) \sigma_z],
\end{split}
\end{equation}
in the direct integral space $L^2(\real,\complex^2)=\int_{\real}^{\oplus} \complex^2 dv$; see \cite{KiWe10}. In this representation, $\mathcal K_0=L^2(\real,dv)$, the operator $\Theta$ is the diagonal multiplication operator given by the function $f_\Theta(v)=\arccos \tanh(\pi v/2)$, and we have $\spec{\Theta}=[0,\pi]$. In order to get to the spectral representation of $\Theta$, we just need to define a measure $\mu$ on Borel sets of $[0,\pi]$ via
\begin{equation}\label{measureQP}
\mu(X) := \int_{f_\Theta^{-1}(X)} dx,
\end{equation}
where $dx$ is the Lebesgue measure on $\real$. The function $f_\Theta:\real\to (0,\pi)$ induces a unitary operator from $L^2(\real)$ onto $L^2([0,\pi],\mu)$, which indeed transforms \eqref{eq:scaledMQP} exactly to \eqref{eq:C2representation}.

To summarize, $\hik^\perp=\{0\}$, and the algebra $\algebra$ is completely characterized by \eqref{dirintHrep}, which is now a genuine direct integral over the full interval $\spec{\Theta}=[0,\pi]$ (the spectrum is purely absolutely continuous), $\hi_\theta=\complex$ for each $\theta\in [0,\pi]$, and $\mu:\mathcal{B}([0,\pi])\to [0,\infty)$ is absolutely continuous with respect to the Lebesgue measure.

Since $\spec{\Theta}=[0,\pi]$, it follows that $sP_1$ and $tP_2$ are coexistent if and only if $s+t \leq 1$, as the region of admissible scalings $(s,t)$ depends only on the largest element of the spectrum of $H$, i.e. the bottom of the spectrum of $\Theta$.

In general the region of admissible pairs $(s,t)$ decreases monotonically with this number. The minimal region, characterised by $s+t\leq 1$, is attained if and only if $\inf \spec{\Theta}=0$. Since $0$ is not an eigenvalue of $\Theta$, this can only happen if the Hilbert space is infinite-dimensional; we give an explicit example of this case below.
\end{example}

\begin{example}\label{ex:copies}
It was noted in Example~\ref{ex:tensor} that $A_1\otimes A_2$ and $B_1 \otimes B_2$ can be coexistent even if the pairs $A_1,B_1$ and $A_2,B_2$ are not. 
We now use Prop.~\ref{prop:scaling} to examine what happens to coexistence of the scaled projections when we consider multiple copies of the effects. Let us take a simple example with $P_1=|\psi_1\rangle\langle\psi_1|$ and $P_2=|\psi_2\rangle\langle\psi_2|$, with $\psi_1,\psi_2\in \hi$ non-parallel unit vectors. Now the operator $H$ is just the scalar $|\langle\psi_1|\psi_2\rangle|^2$, acting in the one-dimensional space spanned by $\psi_1$, so $sP_1\coex tP_2$ if and only if 
\begin{equation}
\label{eq:singlecond}
|\langle\psi_1|\psi_2\rangle|^2\leq\frac{1-s}{s}\frac{1-t}{t} \, .
\end{equation}
We take $n$ copies of these systems, $\hi_{\rm tot}=\hi \otimes \cdots \otimes \hi$ and look at the coexistence of the effects $(sP_1)^{\otimes n}$ and $(tP_2)^{\otimes n}$. 
These are also scaled projections, $(sP_1)^{\otimes n}=s^n P_1^{\otimes n}$ and $(tP_2)^{\otimes n}=t^n P_2^{\otimes n}$. The two projections $P_1^{\otimes n}$ and $P_2^{\otimes n}$ are rank one, the corresponding unit vectors being $\psi_1^{\otimes n}$ and $\psi_2^{\otimes n}$, respectively, so their central element $H$ is just $|\langle\psi_1|\psi_2\rangle|^{2n}$, acting on the one-dimensional subspace spanned by the vector $\psi_1^{\otimes n}$. 
Hence, according to Prop.~\ref{prop:scaling}, $(sP_1)^{\otimes n}\coex (tP_2)^{\otimes n}$ if and only if 
\begin{equation}
\label{eq:multicond}
|\langle\psi_1|\psi_2\rangle|^{2n}\leq\frac{1-s^n}{s^n}\frac{1-t^n}{t^n}.
\end{equation}
\begin{figure}
\begin{center}
\includegraphics{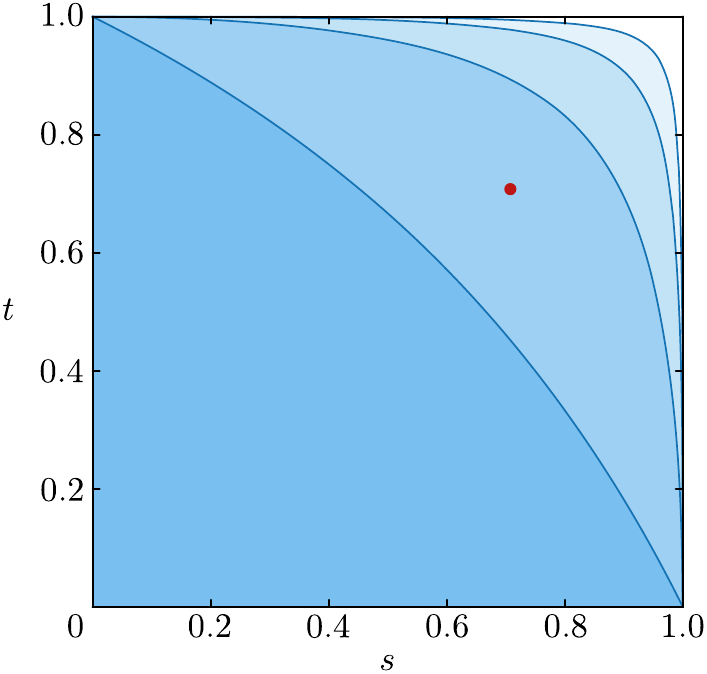}
\end{center}
\caption{Regions of pairs $(s,t)\in[0,1]\times[0,1]$ from Example~\ref{ex:copies} that lead to coexistence effects.
We have depicted $|\langle \psi_1|\psi_2\rangle|=\frac {1}{\sqrt 2}$ for various numbers of copies from $n=1$ (darkest region) to $n=4$ (lightest region). For every pair $(s,t)$ with $s,t<1$ there exists $n$ such that the point lies in admissible region. Example \ref{ex:tensor} is displayed as the red dot in the region of $n=2$ but outside region for $n=1$.}
\label{fig:copies}
\end{figure}
By taking $n$ sufficiently large, any scaling $(s,t)$ with $s,t<1$ leads to a pair of coexistent effects; see Fig.~\ref{fig:copies}.
In particular, for $s=t$ the inequality \eqref{eq:multicond} is equivalent to
\[
s=t\leq\left(\frac{1}{1+|\langle\psi_1|\psi_2\rangle|^n}\right)^{1/n}.
\]
In the case where $|\langle \psi_1|\psi_2\rangle|=\frac {1}{\sqrt 2}$ and $n=1$,  the above bound is $2-\sqrt{2}$ and for $n=2$ the bound is $\sqrt{2/3}$.
This explains the situation in Example \ref{ex:tensor}; the scaling $s=1/\sqrt{2}$ is larger than the bounding value for $n=1$, but nevertheless smaller than that for $n=2$.

\end{example}

\begin{example}\label{ex:dim4}
The existence of pairs of projections with trivial commutation domain $\{0\}$ requires that the dimension of the Hilbert space be even. This is clear from the decomposition, because $\mathcal K \simeq \hik_0\otimes \complex^2$ has even dimension. Here we give an example of a four-dimensional case.

Let $\{\psi_1,\psi_2,\psi_3,\psi_4\}$ be an orthonormal basis for $\hi=\complex^4$.
We denote $P_1=\kb{\psi_1}{\psi_1}+\kb{\psi_2}{\psi_2}$ and $P_2=\fourier P_1 \fourier^\ast$, where $\fourier$ is the finite Fourier transform w.r.t. basis $\{\psi_1,\psi_2,\psi_3,\psi_4\}$.
As matrices they read
\begin{align*}
P_1=\left( \begin{array}{cccc} 1 & 0 & 0 & 0 \\ 0 & 1 & 0 & 0  \\ 0 & 0 & 0 & 0  \\ 0 & 0 & 0 & 0 \end{array} \right) \, \qquad
P_2=1/4 \left( \begin{array}{cccc} 2 & 1+i  & 0 & 1-i \\ 1-i & 2 & 1+i & 0  \\ 0 & 1-i & 2 & 1+i  \\ 1+i & 0 & 1-i & 2 \end{array} \right)
\end{align*}
We now have $\ker \bigl( [P_1,P_2] \bigr)=\{ 0\}$, $\hik=\hi$ and $\hik_0$ is the two-dimensional subspace spanned by $\psi_1$ and $\psi_2$.
Further,
\begin{equation*}
H=P_1P_2P_1|_{\hik_0}=1/4 \left( \begin{array}{cc} 2 & 1+i  \\ 1-i & 2  \end{array} \right) \, , 
\end{equation*}
so that the spectrum of operator $\Theta$ is $\{\pi/4,3\pi/4\}$.

\begin{figure}
\begin{center}
\includegraphics{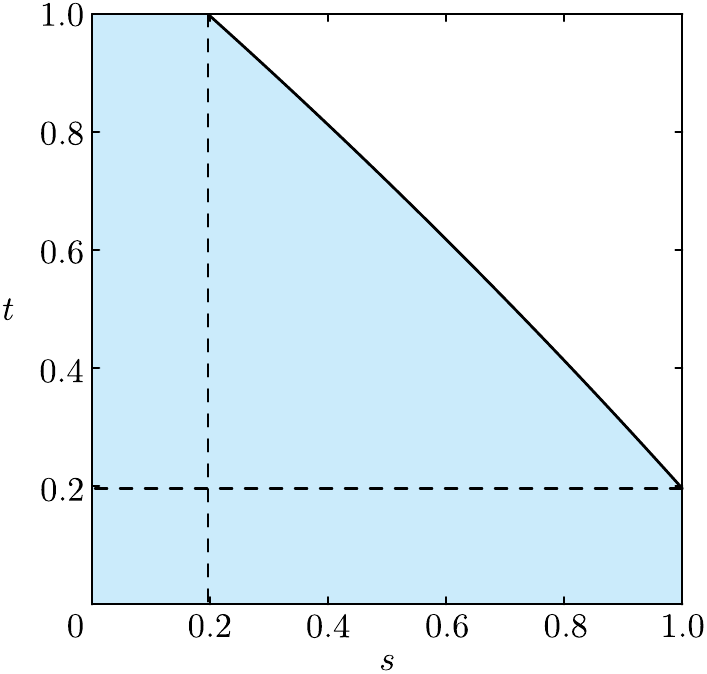}
\end{center}
\caption{Effects $A_s$ and $B_t$ from Example~\ref{ex:dim4} are coexistent for $s$ and $t$ complying with Eq.~(\ref{eq:dim4result}) leading to the shaded region of coexistence bounded by hyperbole (thick line). For choices of $t$ below the dashed line, the two effects are coexistent for any $s$ (similarly for small $s$).}
\label{fig:dim4}
\end{figure}

Let us consider two 1-parameter families of effects,
\begin{align*}
A_s  := s P_1P_2P_1=sCP_1 \, , \qquad B_t := t P_2P_1P_2=t CP_2 \, , 
\end{align*}
with $0 < s,t \leq 1$, and $C$ the central element. From this we immediately see that
\begin{align*}
M_{A_s}(\theta) &= s \cos^2(\theta/2) M_{P_1}(\theta), & M_{B_t}(\theta) &= t \cos^2(\theta/2) M_{P_2}(\theta)
\end{align*}
and we can readily apply \eqref{eq:scaledPPP}. 
We thus obtain condition
\[
\cos^2(\theta/2)\leq\frac{1-s\cos^2(\theta/2)}{s\cos^2(\theta/2)}\frac{1-t\cos^2(\theta/2)}{t\cos^2(\theta/2)}
\]
that has to hold for both $\theta\in\{\pi/4,3\pi/4\}$. 
One can easily check that for $\theta=3\pi/4$ the inequality is trivially fulfilled and therefore we are left with inequality
\begin{equation}
\label{eq:dim4result}
q\leq\frac{1-qs}{qs}\frac{1-qt}{qt},\qquad q\equiv\cos^2\frac{\pi}{8}=\frac{2+\sqrt{2}}{4}.
\end{equation}
This fully characterizes the coexistence region. We note, that by setting $s=1$ the inequality transforms into $t\leq 8(3-2\sqrt{2})/7$ meaning that there are subregions of coexistence within full range of $s$ (or $t$); see Fig.~\ref{fig:dim4}.
\end{example}

\section{Details of the proof of the main result}\label{sec:details}

\subsection{Technical results on coexistence}
We first present some technical results on the set $\mathcal{C}(A,B)$ (defined in \eqref{coexistenteffects}) which are needed later but are also interesting in their own right.

\begin{lemma}\label{comcoex} 
Let $\mathcal{A}\subset \lh$ be a von Neumann algebra having an Abelian commutant $\mathcal{A}'$, and $A,B\in \mathcal{A}$ two coexistent effects. 
Then there exists four effects $G_1,G_2,G_3,G_4\in\mathcal{A}$ that satisfy \eqref{eq:G}.
\end{lemma}

\begin{proof} Since $A$ and $B$ are coexistent, we have $\mathcal{C}(A,B)\neq \emptyset$. 
Each member of the commuting family
\begin{equation*}
\{U^*\,\cdot\,U\mid U\in \mathcal{A}', \, U \text{ unitary }\}
\end{equation*}
of linear, weak-* continuous maps keeps $\mathcal{C}(A,B)$ invariant because $A$ and $B$ commute with each $U$. 
Since $\mathcal{C}(A,B)$ is weak-* compact and convex, the Markov-Kakutani fixed point theorem (see e.g. \cite[Theorem V.20]{MMMPI80}) says that there exists $E\in \mathcal{C}(A,B)$ such that $UE=EU$ for all unitary $U\in \mathcal{A}'$. 
Since each element of a von Neumann algebra is a linear combination of four unitaries from the same algebra, it follows that $E$ commutes with $\mathcal{A}'$. 
By the bicommutant theorem, $E\in \mathcal{A}$.
\end{proof}

The following measurable selection result is essential to the main result of the paper. Recall that $\cal B$ denotes the closed unit ball of $\mathcal{L(H)}$.
\begin{lemma}\label{mselectionlemma}
Let $I$ be a $\sigma$-compact Polish space, and $\mu$ the completion of a Borel measure on $I$ which is finite on compact sets. Suppose that
\begin{align*}
I&\ni x\mapsto A(x)\in \cal B, &  I&\ni x\mapsto B(x)\in \cal B
\end{align*}
are $\mu$-measurable effect-valued functions, such that $A(x)\coex B(x)$ for $\mu$-almost all $x\in I$. Then there exists a family $\{ \G_x\}_{x\in I}$ of four-outcome observables, such that
\begin{itemize}
\item[(a)] $I \ni x\mapsto \G_x(X)\in \cal B$ is $\mu$-measurable for each $X\subseteq \{0,1\}\times \{0,1\}$;
\item[(b)] $\G_x$ is a joint observable for $A(x)$ and $B(x)$ for $\mu$-almost all $x\in I$.
\end{itemize}
\end{lemma}
\begin{proof} Since $\cal B$ is separable, there exists a Borel $\mu$-null set $N$ such that $I'\ni x\mapsto A(x)$ and $I'\ni x\mapsto B(x)$ are Borel functions, where $I'=I\setminus N$. 
Then the sets
\begin{align*}
& \{ (x,E)\in I'\times \cal B\mid A(x)-E\geq 0\},\\
&\{ (x,E)\in I'\times \cal B\mid B(x)-E\geq 0\},\\
&\{ (x,E)\in I'\times \cal B\mid E+\id-A(x)-B(x)\geq 0\},\\
 & \{ (x,E)\in I'\times \cal B\mid E\geq 0\},
\end{align*}
are Borel subsets of $I'\times B$ because the cone of positive operators is closed in $\cal B$. Hence also their intersection
$$\{ (x,E)\in I'\times \cal B\mid E\in \mathcal{C}(A(x),B(x))\}$$ is a Borel subset of $I'\times \cal B$.

Let $\tilde N\subset I$ be a Borel $\mu$-null set such that $A(x)\coex B(x)$ for all $x\in I\setminus \tilde N$. Then $\mathcal{C}(A(x),B(x))\neq \emptyset$ for all $x\in \tilde{I}:=I\setminus (N\cup \tilde N)$. Due to the assumptions on $I$, as well as the fact that $\cal B$ is a Polish space, we can now apply the measurable selection principle of \cite{FTOAII86} (see Theorems 14.3.5 and 14.3.6) to establish the existence of a $\mu$-measurable function $\tilde{I} \ni x\mapsto E(x)\in \cal B$, such that $E(x)\in\mathcal{C}(A(x),B(x))$ for each $x\in \tilde{I}$. Since $\tilde{I}$ is Borel, and $I\setminus \tilde{I}$ is $\mu$-null, this extends to a $\mu$-measurable effect-valued function on $I$ (by defining it to be zero outside $\tilde I$), giving rise to the family of observables $\G_x$ with the desired properties.
\end{proof}

\subsection{Coexistence of decomposable effects in direct integral spaces}\label{dirintsec}

There is an explicit way of formulating the setting of Lemma \ref{comcoex}, since any abelian von Neumann algebra is unitarily equivalent to the algebra of diagonal operators on some \emph{direct integral decomposition} of $\hi$. Our basic reference for this appendix is \cite{FTOAII86}. Even though we only use direct integrals over the real line in the main result, the characterisation of coexistence in direct integral spaces is independently interesting, and so we prove it for the general case (Prop. \ref{dirintcoex} below).

Most of the technical assumptions listed below are needed in order to ensure the existence of certain measurable selections without which the direct integral space does not behave well. Since we explicitly use such a selection below, we need to be careful on details.

Let $I$ be a $\sigma$-locally compact complete separable metric space, $\mu$ the completion of a Borel measure on $I$, and $\hi_x$ a Hilbert space for each $x\in I$. Suppose that to each $\varphi\in \hi$ there exists a function $I\ni x\mapsto \varphi(x)$ such that $\varphi(x)\in \hi_x$ for all $x$,
and
\begin{itemize}
\item[(i)] $x\mapsto \langle \psi(x)|\varphi(x)\rangle$ is $\mu$-measurable integrable function for all $\psi,\varphi\in \hi$, and
\item[(ii)] if $h(x)\in \hi_x$ for all $x\in I$ and $x\mapsto \langle \psi(x)|h(x)\rangle$ is $\mu$-measurable integrable function for each $\psi\in \hi$, then there exists a $\varphi\in \hi$ such that $\varphi(x)=h(x)$ for almost all $x\in I$.
\end{itemize}
Then we denote
$$
\hi = \int_I^{\oplus} \hi_x \, d\mu(x),
$$
and call this the decomposition of $\hi$ into a \emph{direct integral} of Hilbert spaces $\hi_x$. An operator $A\in \lh$ is \emph{decomposable} (with respect to this direct integral) when for each $x\in I$ there exists $A(x)\in \mathcal{L}(\hi_x)$, such that $(A\varphi)(x)=A(x)\varphi(x)$ for almost all $x\in I$.

Note that if $A$ is decomposable, the map $x\mapsto \langle \psi(x)|A(x)\varphi(x)\rangle$ coincides almost everywhere to the measurable and integrable function $x\mapsto \langle \psi(x)|(A\varphi)(x)\rangle$ (but is not required to be itself measurable.) On the other hand, given a family of bounded operators $A(x)$, $x\in I$, such that the essential supremum of $\|A(x)\|$ is finite, and $x\mapsto \langle \psi(x)|A(x)\varphi(x)\rangle$ is measurable for all $\psi,\varphi\in \hi$, it follows from the definition (ii) that there is a decomposable operator $A\in \lh$ such that for all $\varphi\in \hi$ we have $(A\varphi)(x)=A(x)\varphi(x)$ for almost all $x\in I$. A decomposable operator $A$ is positive if and only if $A(x)\geq 0$ for almost all $x\in I$.

We also need the following fact, which says that the ``fibers'' $\hi_x$ can be identified as subspaces of some separable Hilbert space $\mathcal{K}$. That is, there exists a family of isometries $U_x:\hi_x\to \mathcal{K}$, such that $x\mapsto U_x\varphi(x)$ is measurable for each $\varphi\in \hi$, and $x\mapsto U_xS(x)U_x^*$ is measurable for each decomposable operator $A$. Here measurability is with respect to the Borel structures of the norm topology of $\mathcal{K}$ and the ultraweak topology of $\mathcal{L}(\mathcal{K})$. By using separability, it can be shown that this is equivalent to the measurability of all scalar valued maps $x\mapsto \langle \eta |U_x\varphi(x)\rangle$ and $x\mapsto \langle \eta |U_xS(x)U_x^*\eta'\rangle$ for $\eta,\eta'\in \mathcal{K}$ (see \cite[p. 1019]{FTOAII86}).  We fix $\mathcal{K}$, and the isometries $U_x$ for the remainder of this appendix.

The basic result concerning coexistence is the following:
\begin{proposition}\label{dirintcoex} Let $E_1$ and $E_2$ be two decomposable effects on the direct integral Hilbert space $\hi$. Then $E_1$ and $E_2$ are coexistent if and only if $E_1(x)$ and $E_2(x)$ are coexistent effects on $\hi_x$ for $\mu$-almost all $x\in I$. In that case, $E_1$ and $E_2$ have a joint four-outcome observable whose effects are all decomposable.
\end{proposition}
\begin{proof} Let $\cal B$ be the unit ball of $\mathcal L(\mathcal K)$. Assume that $E_1(x)$ and $E_2(x)$ are coexistent for $\mu$-almost all $x\in I$.
Since the maps $I\ni x\mapsto U_xE_1(x)U_x^*\in \cal B$ and $I\ni x\mapsto U_xE_2(x)U_x^*\in B$ are measurable (see the discussion above), it follows from Lemma \ref{mselectionlemma} that there exists a measurable map $I\ni x\mapsto \tilde E(x)\in \cal B$, such that $\tilde E(x)\in \mathcal{C}(U_xE_1(x)U_x^*,U_xE_2(x)U_x^*)$ for almost all $x\in I$. Hence, $E(x):=U_x^*\tilde{E}(x)U_x\in \mathcal{C}(E_1(x),E_2(x))$ for almost all $x\in I$ (since $U_x^*U_x=\id$), and for each $\psi,\varphi\in \hi$, the bounded map
$$I\ni x\mapsto \langle \psi(x)|E(x)\varphi(x)\rangle= \langle U_x\psi(x)|\tilde{E}(x)U_x\varphi(x)\rangle\in \complex$$
is easily seen to be measurable (we can approximate the $\mathcal{K}$-valued measurable functions $x\mapsto U_x\psi(x)$ pointwise with simple functions). Hence, the family $E(x)$, $x\in I$, defines a bounded decomposable operator (note that $\|E(x)\|\leq 1$ for all $x\in I$). Since $E(x)\in \mathcal{C}(E_1(x),E_2(x))$ for almost all $x\in I$, it follows that $E\in \mathcal{C}(E_1,E_2)$. This implies that $E_1$ and $E_2$ are coexistent.

Suppose then that $E_1$ and $E_2$ are coexistent. Now the set of decomposable operators is a von Neumann algebra with Abelian commutant (the latter consisting of diagonal decomposable operators $f(x)\id_{\hi_x}$, where $f\in L^\infty(I,\mu)$); see \cite{FTOAII86}. Hence, by Lemma \ref{comcoex}, there exists a decomposable operator $E\in \mathcal{C}(E_1,E_2)$, (which defines the joint four-outcome observable for $E_1$ and $E_2$). By the definition of $\mathcal{C}(E_1(x),E_2(x))$, and the characterization of decomposable positive operators (see above) this implies that $E(x)\in \mathcal{C}(E_1(x),E_2(x))$ for almost all $x\in I$, that is, $E_1(x)$ and $E_2(x)$ are coexistent for almost all $x\in I$.
\end{proof}

We actually need a slightly different result, in the following setting. Let $\hi'$ be another Hilbert space. Then it is easy to see from the definition that the tensor product $\hi\otimes \hi'$ is also a direct integral, namely
\begin{equation}\label{dirinttensor}
\hi\otimes \hi'=\int_I^\oplus \hi_x\otimes \hi' \, d\mu(x).
\end{equation}
Now if $I\ni x\mapsto A(x)'\in \mathcal{B}(\hi')$ is a strongly (or equivalently, weakly) measurable family with $x\mapsto \|A(x)'\|$ being $\mu$-essentially bounded, then the family
$$x\mapsto \id_{\hi_x}\otimes A(x)'$$ defines a bounded decomposable operator on the above direct integral. This is because for any $\psi,\varphi\in \hi\otimes \hi'$, we have
\begin{equation*}
\langle \psi(x)|(\id_{\hi_x}\otimes A(x)')\varphi(x)\rangle= \langle (U_x\otimes \id_{\hi'}) \psi(x)|(\id_{\hi_x}\otimes A(x)') (U_x\otimes \id_{\hi'})\varphi(x)\rangle,
\end{equation*}
and we can then write $(U_x\otimes \id_{\hi'}) \psi(x)$ in a fixed basis of $\mathcal{K}\otimes \hi'$ with measurable components.

In the following proposition, we let $\mathcal{R}$ denote the set of all decomposable operators in the direct integral \eqref{dirinttensor}, and $\mathcal{R}_0\subset \mathcal{R}$ the set of those $A\in \mathcal{R}$ such that
$$
A(x) = \id_{\hi_x}\otimes A(x)' \text{ for almost all } x\in I,
$$
where $x\mapsto A(x)'$ is a measurable family with $x\mapsto \|A(x)'\|$ being $\mu$-essentially bounded.

\begin{proposition}\label{dirintcoextensor} Let $E_1,E_2\in \mathcal R_0$. Then $E_1$ and $E_2$ are coexistent if and only if $E_1(x)'$ and $E_2(x)'$ are coexistent effects on $\hi'$ for $\mu$-almost all $x\in I$. In that case, $E_1$ and $E_2$ have a joint four-outcome observable whose effects are all in $\mathcal{R}_0$.
\end{proposition}
\begin{proof} It follows from Prop. \ref{dirintcoex} that $E_1$ and $E_2$ are coexistent if and only if $E_1(x)=\id_{\hi_x}\otimes E_1(x)'$ and $E_2(x)=\id_{\hi_x}\otimes E_2(x)'$ are coexistent for almost all $x\in I$, in which case there exists a decomposable operator $E$ with $E(x)\in \mathcal{C}(E_1(x),E_2(x))$ for almost all $x\in I$. Assuming now that this is the case, we fix $\eta_0\in \hi$, and define an isometry $V_x:\hi'\to \hi_x\otimes \hi'$ via
$V_x\varphi =\|\eta_0(x)\|^{-1}\eta_0(x)\otimes \varphi$, and set
$E(x)':=V_x^*E(x)V_x$ for all $x\in I$. Then we have, for each $\varphi,\psi\in \hi'$, that
$$
x\mapsto \langle \varphi |E(x)'\psi\rangle = \|\eta_0(x)\|^{-2}\langle \eta_0(x)\otimes \varphi |E(x) (\eta_0(x)\otimes \psi)\rangle
$$
is measurable (because $x\mapsto \eta_0(x)\otimes \varphi$ defines a vector in the direct integral), so $x\mapsto E(x)'$ is weakly, hence also strongly, measurable. Since $V_x^* \cdot V_x$ is a positive map, and $V_x^* (\id_{\hi_x}\otimes E_i(x)')V_x = E_i(x)'$ for all $x$, it follows that $E(x)'\in \mathcal{C}(E_1(x)',E_2(x)')$ for almost all $x$. Hence, $E_1(x)'$ and $E_2(x)'$ are coexistent for almost all $x$. Moreover, $\id_{\hi_x}\otimes E(x)'$ is clearly an element of $\mathcal{C}(E_1(x),E_2(x))$ for almost all $x$, so the decomposable operator $\tilde{E}\in \mathcal{R}_0$ given by $\tilde{E}(x)=\id_{\hi_x}\otimes E(x)'$ defines a joint four-outcome observable for $E_1$ and $E_2$. Assuming, conversely, that $E_1(x)'$ and $E_2(x)'$ are coexistent for almost all $x$, it trivially follows that $E_1(x)$ and $E_2(x)$ are coexistent for any such $x$.
\end{proof}

Finally, we will make use of the following form of the spectral theorem:
\begin{proposition}\label{spectraltheorem} Let $H$ be a bounded selfadjoint operator in some separable Hilbert space $\hi$. Then there exists a Borel measure $\mu$ on the interval $I=[-\|H\|,\|H\|]\subset \real$, and a family of Hilbert spaces $\{\hi_h\}_{h\in I}$, such that $\hi$ is the direct integral space
$$
\hi = \int_{I}^{\oplus} \hi_h \,d\mu(h),
$$
and $H$ is the diagonal multiplication operator $H(h)=h\id_{\hi_h}$. 
Moreover, the measure class\footnote{Two Borel measures belong to the same measure class if they have the same null sets.} of $\mu$ is uniquely determined, and the spectrum $\sigma(H)$ of $H$ coincides with the support of $\mu$, that is,
\begin{equation}\label{spectrumcharacterization}
h\in \sigma(H) \text{ if and only if }\,\mu((h-\epsilon,h+\epsilon))>0 \text{ for all } \epsilon>0.
\end{equation}
The abelian von Neumann algebra generated by $H$ coincides with the algebra of diagonal decomposable operators in this direct integral representation.
\end{proposition}

\section{Discussion}

We have presented a general condition for the coexistence of a pair of effects belonging to some von Neumann algebra generated by two projections. This condition relies on a decomposition of given effects to blocks of two-by-two matrices and checking their coexistence, for which conditions are already known. By comparing our condition to the condition of generalized infimum \eqref{eq:ginf} we conclude that although \eqref{eq:ginf} provides a good sufficient condition, it is not a necessary condition. That is where our result proves useful.
In a particular exemplary case of scaled projections we found our condition to provide another advantage. Although the coexistence is in this case equivalent to \eqref{eq:ginf}, our results allow uncomplicated analysis.

Since the main result of this paper characterizes the coexistence of effects in an algebra generated by two projections, one can ask when a given pair of effects belongs to such algebra. Another natural question is if algebras generated by more than two projections also allow similar analysis.

Interestingly, the situation becomes rather intractable even with three projections. As shown in \cite{Davis55}, the full algebra $\lh$ corresponding to any separable Hilbert space is generated by three projections. This emphasises the special character of the algebra of two projections, which does not coincide with $\lh$ if $\dim\hi\geq 3$.
Indeed, if $\algebra=\lh$, then the central element $C$ defined in \ref{eq:centralC} must be a scalar multiple of the identity, but a short argument (p. 970 in \cite{Davis55}) shows that this cannot be true if $\dim\hi\geq 3$. In general, the structure of algebras generated by three or more projections is known only in some special cases \cite{Vasilevski98}.

After this negative result we can still ask whether every pair of effects belongs to the algebra generated by some pair of projections.
This would mean that we could always apply Prop.~\ref{prop:compatibilityvNm}, the  remaining problem would just be to find the projections in question.
However, it is easy to give examples of pair of effects that do not belong to any algebra generated by two projections.
For instance, fix two orthonormal bases $\{\phi_j\}_j$ and $\{\psi_k\}_k$ satisfying $\langle \phi_j|\psi_k\rangle \neq 0$ for all $j,k$.
Let $A$ and $B$ be two full rank effects that are diagonal with nondegenerate eigenvalues in the bases $\{\phi_j\}_j$ and $\{\psi_k\}_k$, respectively.
Then the spectral projections $|\phi_j\rangle\langle\phi_j|$ and $|\psi_k\rangle\langle \psi_k|$ are all in the von Neumann algebra generated by these effects, so also each rank one operator $|\phi_j\rangle\langle \psi_k|$ is in the algebra. 
Since any bounded operator is a weak-* limit of linear combinations of such operators, the von Neumann algebra generated by $A$ and $B$ is the whole $\lh$. 
Therefore, if $\dim\hi\geq 3$ there cannot be projections $P_1$ and $P_2$ such that $A,B\in\algebra$. 

We conclude that the results presented in this paper generalize the qubit coexistence characterization to a large set of effect pairs in all dimensions, and the introduced machinery provides an effortless tool to decide coexistence. 
However, we are still quite far from a complete characterization of all coexistent pairs of effects.

\section{Acknowledgement}
T.H. acknowledges support from the Academy of Finland (grant no. 138135).
J.K. and D.R. acknowledge support from European CHIST-ERA/BMBF project CQC. J.K. acknowledges additional support from the EPSRC project EP/J009776/1.

\end{document}